\documentclass{elsarticle}  
\usepackage{amssymb}
\usepackage{amsfonts}
\usepackage{amsbsy}
\usepackage{amsmath}

\usepackage{delarray,multirow,array}
\usepackage[portuguese]{babel}
\usepackage[latin1]{inputenc}
\usepackage{color}

\newtheorem{theorem}{Theorem}[section]
\newtheorem{definition}[theorem]{Definition}

\newtheorem{corollary}[theorem]{Corollary}

\newtheorem{example}[theorem]{Example}

\def\F{{\mathbb{F}}}

\newenvironment{proof}{\textbf{Proof:}}{\hspace*{\fill}
\nolinebreak\hspace*{\fill}$\Box$\newline\vspace{1mm}}

\newcommand{\im}{\textnormal{im}}

\newcommand{\st}{\ensuremath{ \ | \ }}

\title{\LARGE \bf
Superregular matrices and applications to convolutional codes
}

\author[rvt]{P. J. Almeida }
\ead{palmeida@ua.pt}

\author[rvt]{D. Napp  }
\ead{diego@ua.pt}

\author[rvt]{R.Pinto \corref{cor1}}
\ead{raquel@ua.pt}

\cortext[cor1]{Corresponding author}

\address[rvt]{Department of Mathematics, University of Aveiro, Campus Universit\'ario de Santiago, 3810--193 Aveiro, Portugal.}

\fntext[Pi]{This work was supported by Portuguese funds through the CIDMA - Center for Research and Development in Mathematics and Applications, and the Portuguese Foundation for Science and Technology (FCT-Funda\c{c}\~ao para a Ci\^encia e a Tecnologia), within project PEst-UID/MAT/04106/2013.
}

\begin{document}
\begin{abstract}
The main results of this paper are twofold: the first one is a matrix theoretical result. We say that a matriz is superregular if all of its minors that are not trivially zero are nonzero. Given a $a\times b$, $a\geq b$, superregular matrix over a field, we show that if all of its rows are nonzero then any linear combination of its columns, with nonzero coefficients, has at least $a-b+1$ nonzero entries. Secondly, we make use of this result to construct convolutional codes that attain the maximum possible distance for some fixed parameters of the code, namely, the rate and the Forney indices. These results answer some open questions on distances and constructions of convolutional codes posted in the literature \cite{GRS2006,mceliece98}.
\end{abstract}

\begin{keyword}
convolutional code \sep Forney indices \sep optimal code \sep superregular matrix 

2000MSC: 94B10, 15B33
\end{keyword}

\maketitle

\section{Introduction}

Several notions of superregular matrices (or totally positive) have appeared in different areas of mathematics and engineering having in common the specification of some properties regarding their minors \cite{Ando1987,CIM1998,Gan59a,Pinkus2009,Roth1985}. In the context of coding theory these matrices have entries in a finite field $\F$ and are important because they can be used to generate linear codes with good distance properties. A class of these matrices, which we will call {\em full superregular}, were first introduced in the context of block codes. A full superregular matrix is a matrix with all of its minors different from zero and therefore all of its entries nonzero. It is easy to see that a matrix is full superregular if and only if any $\F$-linear combination of $N$ columns (or rows) has at most $N-1$ zero entries. For instance, Cauchy and nonsingular Vandermonde matrices are full superregular. 
It is well-known that a systematic generator matrix $G = [I\ | \ B]^\top$  generates a maximum distance separable (MDS) block code if and only if $B$ is full superregular, \cite{Roth1989}.\\[1ex]

Convolutional codes are more involved than block codes and, for this reason, a more general class of superregular matrices had to be introduced. A lower triangular matrix $B$ was defined to be superregular if all of its minors, with the property that all the entries in their diagonals are coming from the lower triangular part of $B$, are nonsingular, see \cite[Definition 3.3]{GRS2006}. In this paper, we call such matrices {\em LT-superregular}. Note that due to such a lower triangular configuration the remaining minors are necessarily zero. Roughly speaking, superregularity asks for all minors that are possibly nonzero, to be nonzero. In \cite{GRS2006} it was shown that LT-superregular matrices can be used to construct convolutional codes of rate $k/n$ and degree $\delta$ that are strongly MDS provided that $(n-k)\mid \delta$. This is again due to the fact that the combination of columns of these LT-superregular matrices ensures the largest number of possible nonzero entries for any $\F$-linear combination (for this particular lower triangular structure). In other words, it can be deduced from \cite{GRS2006} that a lower triangular matrix $B=[b_0 ~ b_1 \dots b_{k-1}]\in \F^{n \times n}$, $b_i$ the columns of $B$, is LT-superregular  if and only if for any $\F$-linear combination of columns $ b_{i_1}, b_{i_2}, \dots , b_{i_N} $ of $B$, with $i_j < i_{j+1}$, then $wt(b) \geq wt(b_{i_1})-   N +1 = (n-i_1)- N +1$. \\[1ex]

It is important to note that in this case due to this triangular configuration it is hard to come up with an algebraic construction of LT-superregular matrices.  There exist however two general constructions of these matrices \cite{ANP2013,GRS2006,HuSmTr2008} although they need large field sizes. Unfortunately, LT-superregular matrices allow to construct convolutional codes with optimal distance properties only for certain given parameters of the code. This is because the constant matrix associated to a convolutional code have, in general, blocks of zeros in its lower triangular part. Hence, in order to construct convolutional codes with good distance properties for any set of given parameters a more general notion of superregular matrices needs to be introduced. It is the aim of this paper to do so by generalizing the notion of superregularity to matrices with any structure of zeros. To this end we introduce the notion of \emph{nontrivial} minor (i.e., at least one term in the summation of the Leibniz formula for the determinant is nonzero). Hence, a matrix will be called {\em superregular} if all of its nontrivial minors are nonzero. This notion naturally extends the previous notions of superregularity as they have all of its possible nonzero minors different from zero.\\

A key result in this paper is that any $\F$-linear combination of columns of a superregular matrix have the largest possible number of nonzero components (to be made more precise in Section \ref{sec:superregular}). This is a general matrix theoretical result and it stands in its own right. As an application, we will show that this result will ensure that any convolutional code associated to a superregular matrix have the maximum possible distance. \\

In \cite{GRS2006,mceliece98} it was proved that the distance of a convolutional code with rate $k/n$ and different Forney indices $\nu_1 < \dots < \nu_\ell$ is upper bounded by $n(\nu_1+1)-m_1+1$ where $m_1$ is the multiplicity of the Forney index $\nu_1$. Whether this bound was optimal or not was left as an open question. In this work we show that it is indeed optimal by presenting a class of convolutional codes that achieve such a bound. In the particular case that the given Forney indices have two consecutive values, say $\nu$ and $\nu +1$, then our construction yields a new class of (strongly) MDS convolutional codes. 

\section{Convolutional codes}\label{sec:intro}

In this section we recall basic material from the theory of convolutional codes that is relevant to the presented work. In this paper we consider convolutional codes constituted by codewords having finite support.

Let $\mathbb{F}$ be a finite field  and $\mathbb{F}[z]$ the ring of polynomials with coefficients in $\mathbb{F}$.
A \emph{(finite support) convolutional code} $\mathcal{C}$ of rate $k/n$ is an $\mathbb{F}[z]$-submodule of $\mathbb{F}[z]^{n}$, where $k$ is the rank of $\mathcal{C}$ (see \cite{RosenthalS99}). The elements of $\mathcal{C}$ are called \emph{codewords}.

A full column rank matrix $G(z) \in \mathbb{F}[z]^{n \times k}$ whose columns constitute a basis for $\mathcal{C}$ is called an \emph{encoder} of $\mathcal{C}$.
So,
\begin{align*}
  \mathcal{C}
    & = \im_{\mathbb{F}[z]} G(z)  \nonumber
     = \left\{
          v(z) \in \mathbb{F}[z]^{n}
          \st
          v(z) = G(z) u(z)
          \ \text{with} \
          u(z) \in \mathbb{F}[z]^{k}
        \right\}.
\end{align*}

Convolutional codes of rate $k/n$ are linear devices which map a sequence of $k$-dimensional information words $u_0, u_1, \dots, u_\epsilon$ (expressed as $u(z)= \sum_{i=0}^\epsilon u_i z^i $), into a sequence of $n$-dimensional codewords $v_0, v_1, \dots, v_\gamma$ (written as $v(z)= \sum_{i = 0}^\gamma v_i z^i  $). In this sense it is the same as block codes. The difference is that convolutional encoders have a internal ``storage vector" or ``state vector". Consequently, convolutional codes are often characterized by the code rate and the structure of the storage device.  \\

The $j$-th \emph{column degree} of $G(z)=[g_{ij}(z)]\in \mathbb{F}[z]^{n \times k}$ (also known as constraint length of the $j$-th input of the matrix $G(z)$, see \cite{JZ1999}) is defined as
$$
\nu_j = \max_{1\leq i    \leq n } \deg g_{ij}(z)
$$
the \emph{memory} $m$ of the polynomial encoder as the maximum of the columns degrees, that is,
$$
m= \max_{1\leq j    \leq k } {\nu_j}
$$
and the \emph{total memory}  (or overall constrain length) as the sum of the constraint length
$$
\nu= \sum_{1\leq j \leq k} \nu_j.
$$

The encoder $G(z)$ can be realized by a linear sequential circuit consisting of $k$ \emph{shift registers}, the $j$-th of length $\nu_j$, with the outputs formed as sums of the appropriate shift registers contents.

Two full column rank matrices $G_{1}(z),G_{2}(z) \in \mathbb{F}[z]^{n \times k}$ are said to be equivalent encoders if $\im_{\mathbb{F}[z]} G_{1}(z) =   \im_{\mathbb{F}[z]} G_{2}(z)$, which happens if and only if there exists a unimodular matrix $U(z) \in \mathbb{F}[z]^{k \times k}$ such that $G_{2}(z) = G_{1}(z) U(z)$ \cite{JZ1999,RosenthalS99}.

Among the encoders of the code, the column reduced are the ones with smallest sum of the column degrees.

\begin{definition}
Given a matrix $G(z)=[g_{ij}(z)]\in \mathbb{F}[z]^{n \times k}$ with column degrees $\nu_1, \dots, \nu_k$ let $G^{hc}$ (hc stands for highest coefficient) be the constant matrix whose $(i,j)$-entry is the coefficient of degree $\nu_j$ if $\deg g_{ij}=\nu_j$ or zero otherwise. We say that $G(z)$ is \emph{column reduced} if $G^{hc}$ is full column rank.
\end{definition}

It was shown by Forney \cite{forney75} that two equivalent column reduced encoders have the same column degrees up to a permutation. For this reason such degrees are called the \emph{Forney indices} of the code, see \cite{mceliece98}. The number of Forney indices with a certain value $\nu$ is called the multiplicity of $\nu$. The \emph{degree} of a convolutional code is the sum of the Forney indices of the code.

%
%

\begin{definition}
An important distance measure for a convolutional code $\mathcal{C} $ is the
\textit{distance} $ \mathrm{dist}(\mathcal{C})$ defined as
\[
 \mathrm{dist}(\mathcal{C}):= \left\{ \mathrm{wt}(v(D)) \ \ | \ \
  v(D) \in \mathcal{C} \quad \text{and} \quad
  v(D)\neq \vec{0} \right\},
\]
where $\mathrm{wt}(v(D))$ is the Hamming weight  of a polynomial vector
$$v(D)=\sum\limits_{i \in \mathbb N} v_i D^i
\in \mathbb F[D]^n,$$  defined as $${\rm wt}(v(D))=\sum\limits_{i
\in \mathbb N} {\rm wt}(v_i),$$ where ${\rm wt}(v_i)$ is
the number of nonzero components of $v_i$. \end{definition}

In~\cite{RosenthalS99}, Rosenthal and Smarandache showed that the distance
of a convolutional code of rate $k/n$ and degree $\delta$ must be upper bounded by
\begin{eqnarray}\label{eq1}
   \mathrm{dist}(\mathcal{C})\leq (n-k)\left(\left\lfloor
      \frac{\delta}{k}\right\rfloor + 1\right)+\delta +1.
\end{eqnarray}
This bound was called the  \textit{generalized Singleton bound}
since it generalizes in a natural way the Singleton
bound for block codes (when $\delta=0$). A convolutional code of rate $k/n$ and degree $\delta$  with its distance equal to
the generalized Singleton bound was called a \textit{maximum distance
  separable} (MDS) code~\cite{RosenthalS99}.  It was also observed in~\cite{mceliece98,RosenthalS99} that  if $\mathcal C$ is MDS,
then its set of Forney indices must have $\xi:= k (\left\lfloor \frac{\delta}{k} \right\rfloor +1) -\delta$ indices of value $\left\lfloor \frac{\delta}{k} \right\rfloor$ and $k-\xi$ indices of value $\left\lfloor \frac{\delta}{k} \right\rfloor + 1$ (this set of indices are called in the literature ``generic set of column indices" or ``compact"). Few algebraic constructions of MDS convolutional codes are known, see \cite{Smarandache2001,NaRo2015}. The particular case where $(n-k)$ divides $\delta$ was investigated in \cite{GRS2006}. Note that in this case all the Forney indices of a MDS convolutional code are equal. It is the aim of this paper to study the distance properties of convolutional codes of given rate and \emph{any} set of Forney indices. Equivalents bounds of the distance of these codes were independently given in \cite{RosenthalS99} and in \cite{mceliece98}.

\begin{theorem}\cite{RosenthalS99} \label{ubk1D}
Let $\cal C$ be a convolutional code with rate $k/n$ and different Forney indices $\nu_1 < \dots < \nu_{\ell}$  with corresponding multiplicities $m_1, \dots, m_{\ell}$. Then the distance of $\mathcal{C}$ must satisfy
$$
\mathrm{dist}(\mathcal{C}) \leq n(\nu_1 + 1) - m_1 +1.
$$
\end{theorem}

A convolutional code of rate $k/n$ with different Forney indices $\nu_1< \dots < \nu_{\ell}$ and with corresponding multiplicities $m_1, \dots, m_{\ell}$ and distance $n(\nu_1 + 1) - m_1 +1$ is said to be an \emph{optimal} $(n,k,\nu_1,m_1)$ convolutional code. Note that a convolutional code of rate $k/n$ and degree $\delta$ is MDS if and only if is an optimal $(n,k,\left\lfloor \frac{\delta}{k} \right\rfloor, k (\left\lfloor \frac{\delta}{k} \right\rfloor +1) -\delta)$ convolutional code.\\

It was left as an open question whether there always exist optimal $(n,k,\nu_1,m_1)$ convolutional codes for all rates and Forney indices $\nu_1\leq \dots \leq \nu_{k}$. In the next section, we consider a special class of matrices that will allow us to exhibit convolutional codes with this property.

\section{Superregular Matrices}\label{sec:superregular}

 In this section, we recall some pertinent definitions on superregular matrices and introduce a new construction of superregular matrices that we will use to obtain MDS convolutional codes. Such matrices have some similarities with the ones introduced in \cite{ANP2013}. They have similar entries and, therefore, some properties are the same, even if the structure of these new matrices is different.

Let $\mathbb{F}$ be a field, $A=[\mu_{i\ell}]$ be a square matrix of order $m$ over $\mathbb{F}$ and $S_m$ the symmetric group of order $m$. The determinant of $A$ is given by
\[\mid A\mid =\sum_{\sigma\in S_m} (-1)^{\mbox{sgn}(\sigma)}\mu_{1\sigma(1)}\cdots\mu_{m\sigma(m)}.\]
Whenever we use the word {\em term}, we will be considering one product of the form $\mu_{1\sigma(1)}\cdots\mu_{m\sigma(m)}$, with $\sigma\in S_m$, and the word {\em component} will be reserved to refer to each of the $\mu_{i\sigma(i)}$, with $1\leq i\leq m$ in a term. Denote $\mu_{1\sigma(1)}\cdots\mu_{m\sigma(m)}$ by $\mu_\sigma$.

A {\em trivial term} of the determinant is a term $\mu_\sigma$, with at least one component $\mu_{i\sigma(i)}$ equal to zero. If $A$ is a square submatrix of a matrix $B$ with entries in $\mathbb{F}$, and all the terms of the determinant of $A$ are trivial, we say that $\mid A\mid $ is a {\em trivial minor} of $B$ (if $B=A$ we simply say that $\mid A\mid $ is a trivial minor). We say that a matrix $B$ is {\em superregular} if all its nontrivial minors are different from zero.


In the next theorem we study the weight of vectors belonging to the image of a superregular matrix.

\begin{theorem}\label{SRZeros}
Let $\mathbb{F}$ be a field and $a,b\in\mathbb{N}$, such that $a\geq b$ and $B\in\mathbb{F}^{a\times b}$. Suppose that $u=[u_i]\in\mathbb{F}^{b\times 1}$ is a column matrix such that $u_i\neq 0$ for all $1\leq i\leq b$. If $B$ is a superregular matrix and every row of $B$ has at least one nonzero entry then $\mathrm{wt}(Bu)\geq a-b+1$.
\end{theorem}
\begin{proof}
Suppose that $\mathrm{wt}(Bu)\leq a-b$, then there exists a square submatrix of $B$ of order $b_1=b$, say $B_1$, such that $B_1 u=0$,  and so $\mid B_1\mid =0$, i. e., the columns of $B_1$ are linearly dependent. Since $B$ is superregular, $\mid B_1\mid$ is a trivial minor. By hypothesis $u_i\neq 0$, for all $1\leq i\leq b$, which implies that every row of $B_1$ must have at least two nonzero entries. On the other hand,
$B_1$ may have some of its columns identically equal to zero.

Using the fact that $B_1$ is also superregular, we are going to show that there exists, up to permutation of rows and columns, a square submatrix $B_2$ of $B_1$ of order $b_2$, with $b_2<b_1$, such that $B_2\widetilde{u}=0$, where $\widetilde{u}$ is a column matrix with $b_2$ rows whose entries are elements of $u$. Therefore, $\mid B_2\mid$ is a trivial minor which implies that the columns of $B_2$ are linearly dependent. Also every row of $B_2$ will have at least two nonzero entries. But then, proceeding in this way, we would obtain an infinite sequence $B_1, B_2, B_3, \dots$ of square matrices of orders $b_1, b_2, b_3, \dots$, respectively, with $0<\dots<b_3<b_2<b_1$  all having at least two nonzero entries in every row. Of course, this cannot happen, hence, $\mathrm{wt}(Bu)\geq a-b+1$. This is an application of the infinite descent method of Fermat.

Since $u_i\neq 0$, for all $1\leq i\leq b$, if some of the columns of $B_1$ are identically equal to zero, then the remaining columns are still linearly dependent. Let $\overline{B}$ be the matrix formed by the columns of $B_1$ with at least one nonzero entry and let $\widehat{B}$ be a square submatrix of $\overline{B}$ with the same number of columns. Denote by $m$ the order of $\widehat{B}$. Clearly $m\leq b_1$.

Let $t$ be the dimension of the subspace generated by the columns of $\widehat{B}$. Then $\widehat{B}$ has a $t\times t$ submatrix whose columns are linearly independent. Therefore, its determinant is nonzero and $t<m$. After an adequate permutation of the rows and columns of $\widehat{B}$ we may express the minor $\mid \widehat{B}\mid$ as $\mid \widehat{B} \mid=\pm\mid M \mid$, where
\[M=[\mu_{i\,j}]=
\left [\;\begin{array}{c|c}
 & \\
\qquad\widetilde{B}\qquad & \qquad C\qquad\\
 & \\ \hline
  & \\
\qquad R\qquad & \qquad Z\qquad\\
  &
\end{array}\;\right ],
\]
and where $\widetilde{B}$ is a $t\times t$ nonsingular matrix with nonzero entries in its principal diagonal, i. e.  with $\mu_{i\,i}\neq 0$, for $1\leq i\leq t$, $C$ is a $t\times (m-t)$ matrix, $R$ is a $(m-t)\times t$ matrix, $Z$ is a $(m-t)\times (m-t)$ matrix.

Using the superregularity of $\widehat{B}$, we are going to show that the matrix $M$ has a well defined structure of zeros in its entries.

For any $t+1\leq i_0\leq m$ and any $t+1\leq j_0\leq m$ define $V_{i_0\,j_0}=[v_{i\,j}]$ to be the square $(t+1)\times (t+1)$ matrix formed by $\widetilde{B}$, the $i_0-t$ row of $R$, the $j_0-t$ column of $C$ and the entry $(i_0-t\,j_0-t)$ of $Z$, i. e.

\begin{equation}\label{vij}
V_{i_0\,j_0}=[v_{i\,j}]\;\;\mbox{where}\;\;v_{i\,j}=\left\{\begin{array}{ll}
\mu_{i\,j} & \mbox{if}\;\; 1\leq i,j\leq t\\
\mu_{i_0\,j} & \mbox{if}\;\; i=t+1\;\;\mbox{and}\;\; 1\leq j\leq t\\
\mu_{i\,j_0}  & \mbox{if}\;\; j=t+1\;\;\mbox{and}\;\; 1\leq i\leq t\\
\mu_{i_0\,j_0}  & \mbox{if}\;\; i=j=t+1.
\end{array} \right.
\end{equation}

First, we will show that $Z=0$.

Let  $t+1\leq i_0\leq m$ and $t+1\leq j_0\leq m$ and consider the matrix $V_{i_0\,j_0}$ defined in (\ref{vij}). By the definition of $t$, the columns of $V_{i_0\,j_0}$ are linearly dependent, hence $\mid V_{i_0\,j_0}\mid=0$. Since $\mid V_{i_0\,j_0}\mid$ is a minor of $B_1$ and $B_1$ is superregular, $\mid V_{i_0\,j_0}\mid$ must be a trivial minor. Therefore, the term $v_\sigma$ with $\sigma(i)=i$ is trivial. But $v_{i\,i}=\mu_{i\,i}\neq 0$ for all $1\leq i\leq t$ because these are the entries in the main diagonal of $\widetilde{B}$. Therefore $\mu_{i_0\,j_0}=v_{t+1\,t+1}=0$. This allows us to conclude that $Z=0$.

\vspace{7mm}
Now, we will construct, recursively, three sequences of sets $D_0, D_1, \dots, D_\nu$ and $E_0, E_1, \dots, E_\nu$  and $F_0, F_1, \dots, F_\nu$, where $\nu$ is an integer. 

Let
\begin{equation}\label{defio}
\left.\begin{array}{ll}
F_0=\{1, 2, \dots, t\}\; & \mbox{and}\;\;D_0=E_0=\{t+1, t+2, \dots m\};\\
\mbox{For}\;\;1\leq\lambda\leq\nu, & \\
i\in D_\lambda\;\; & \mbox{if}\;\; i\in F_{\lambda-1}\;\;\mbox{and exists}\;\;i_0\in D_{\lambda-1}\;\;\mbox{such that}\;\;\mu_{i_0\,i}\neq 0;\\
j\in E_\lambda\;\; & \mbox{if}\;\;j\in F_{\lambda-1}\;\;\mbox{and exists}\;\;j_0\in E_{\lambda-1}\;\;\mbox{such that}\;\;\mu_{j\,j_0}\neq 0;\\
k\in F_\lambda\;\; & \mbox{if}\;\;k\in F_{\lambda-1},\; k\notin D_\lambda\;\;\mbox{and}\;\;k\notin E_\lambda.
\end{array} \right \}
\end{equation}
In particular, the set $D_1$ will be the the set formed by the indices of the columns of $R$ that have at least one nonzero entry and $E_1$ will be the set formed by the indices of the rows of $C$ with at least one nonzero entry.

Let $\lambda\in\{1, 2, \dots, \nu\}$. From (\ref{defio}), we immediately have
\begin{equation}\label{zerosR1}
\mbox{if}\;\;i_0\in D_{\lambda-1}\;\;\mbox{and}\;\;i_1\in (F_{\lambda-1}\setminus D_\lambda)\;\;\mbox{then}\;\;\mu_{i_0\,i_1}=0.
\end{equation}
and
\begin{equation}\label{zerosC1}
\mbox{if}\;\;j_0\in E_{\lambda-1}\;\;\mbox{and}\;\;j_1\in (F_{\lambda-1}\setminus E_\lambda)\;\;\mbox{then}\;\;\mu_{j_1\,j_0}=0.
\end{equation}

Let $d_\lambda$, $e_\lambda$ and $f_\lambda$ be the cardinalities of the sets $D_\lambda$, $E_\lambda$ and $F_\lambda$, respectively, and $f_0=t$. Define $\nu$ to be the smallest positive integer for which
\begin{equation}\label{m-t}
m-t\geq \min\{d_\nu,e_\nu,f_\nu\}.
\end{equation}
Observe that one or two sets of $D_\nu$, $E_\nu$ or $F_\nu$ may be empty sets, but, since $m>t$ and $f_{\nu-1} > m-t$, all the other sets of the three sequences are nonempty.

Let us assume that
\begin{equation}\label{empty}
D_\lambda\cap E_\lambda=\emptyset,\;\;\mbox{for any}\;\;\lambda\in\{1, 2, \dots, \nu\},
\end{equation}
and that, for any $\lambda\in\{1,\dots,\nu\}$,
\begin{equation}\label{zeros2}
\mbox{if}\;\; i\in D_\lambda\;\;\mbox{and}\;\;j\in E_\lambda\;\;\mbox{then}\;\;\mu_{i\,j}=0.
\end{equation}
We will prove (\ref{empty}) and (\ref{zeros2}) later.

\vspace{.2cm}

Since $F_{\lambda-1}=D_\lambda\cup E_\lambda\cup F_\lambda$ and, from (\ref{defio}) and (\ref{empty}), $D_\lambda$, $E_\lambda$ and $F_\lambda$ are pairwise disjoint. we have
\begin{equation}\label{fed}
f_{\lambda-1}=d_\lambda+e_\lambda+f_\lambda,
\end{equation}
\begin{equation}\label{f-d}
F_{\lambda-1}\setminus D_\lambda=E_\lambda\cup F_\lambda,
\end{equation}
and
\begin{equation}\label{f-e}
F_{\lambda-1}\setminus E_\lambda=D_\lambda\cup F_\lambda.
\end{equation}

From (\ref{zerosR1}), (\ref{zerosC1}), (\ref{f-d}), (\ref{f-e}) and since $Z=0$, we obtain that if $i\in D_0$ and $j\in E_0$ then the $i$-th row of $M$ has at most $t-(f_1+e_1)=d_1$ nonzero entries and the $j$-th column of $M$ has at most $t-(f_1+d_1)=e_1$ nonzero entries.

Now, given $i\in D_{\lambda}$, for $\lambda\in\{1, 2, \dots, \nu-1\}$, the $i$-th row of $M$ has zeros in all entries $(i,j)$ with $j\in E_{\lambda+1}\cup F_{\lambda+1}$, by (\ref{zerosR1}) and (\ref{f-d}), and in all entries $(i,j)$ with $j\in E_{\lambda}$. We show next that $\mu_{ij}=0$ for $i \in D_{\lambda}$ and $j \in E_\kappa$ for any $1\leq\kappa<\lambda$. If $\kappa = \lambda - 1$, since by (\ref{f-e}) $i \in F_{\lambda-1} \backslash E_{\lambda}$, then by (\ref{zerosC1}) $\lambda_{ij}=0$. If $\kappa < \lambda-1$, since $i \in D_{\lambda}$, then $i \in F_{\kappa + 1}$, so by (\ref{zerosC1}) and (\ref{f-e}), as $j \in E_{\kappa}$, we have $\mu_{ij}=0$. Thus, since $F_{\lambda + 1}, E_{\lambda+1}, E_{\lambda}, \dots, E_1$ are pairwise disjoint,  the $i$-th row of $M$ has at most
$t-(e_{\lambda+1}+f_{\lambda+1}+e_\lambda+e_{\lambda-1}+\cdots+e_1 )$ nonzero entries. Using (\ref{fed}) a few times, we conclude that the number of nonzero entries of the $i$-th row of $M$ is at most
\[d_1+\cdots+d_\lambda+d_{\lambda+1}.\]

Similarly, if $j\in E_\lambda$ , for $\lambda \in \{1,2, \dots, \nu - 1\}$, then the $j$-th column of $M$ has at most
\[e_1+\cdots+e_{\lambda}+e_{\lambda+1}\]
nonzero entries.

Finally, the $i$-th row of $M$, with $i\in D_\nu$, has zeros in all entries $(i,j)$, with $j\in E_{\nu}$, by (\ref{zeros2}), and in all entries $(i,j)$ with $j\in E_\kappa$, with $1\leq\kappa<\nu$, by (\ref{zerosC1}) and (\ref{f-e}). Hence, the $i$-th row of $M$ has at most $t-(e_\nu+e_{\nu-1}+\cdots+e_1)$ nonzero entries. Using (\ref{fed}), we conclude that the number of nonzero entries of the $i$-th row of $M$ is at most
\[d_1+d_2+\cdots+d_\nu+f_\nu.\]

By a similar reasoning, we conclude that the $j$-th column of $M$, with $i\in D_\nu$, has at most
\[e_1+e_2+\cdots+e_\nu+f_\nu\]
nonzero entries.

Permuting the rows of $M$ we obtain a matrix $\bar M$ such that:
\begin{itemize}
\item the last $m-t$ rows remain unchanged;
\item the rows of $M$ in $D_1$ will become the rows $m-t -1, \dots, m-t-d_1$ in $\bar M$;
\item for $\lambda = 2, \dots, \nu$, the rows of $M$ in $D_{\lambda}$ will become the rows $m-t-1-\sum_{i=1}^{\lambda - 1} d_i, \dots, m-t-\sum_{i=1}^{\lambda} d_i$.
\end{itemize}
Applying to $\bar M$ the following column permutations we obtain a matrix $N$ such that:
\begin{itemize}
\item the last $m-t$ columns remain unchanged;
\item the columns of $\bar M$ in $E_1$ will become the columns $m-t -1, \dots, m-t-e_1$ in $N$;
\item for $\lambda = 2, \dots, \nu$, the columns of $\bar M$ in $E_{\lambda}$ will become the columns $m-t-1-\sum_{i=1}^{\lambda - 1} e_i, \dots, m-t-\sum_{i=1}^{\lambda} e_i$.
\end{itemize}

Thus, the matrix $N$ satisfies the following properties:
\begin{enumerate}
\item its last $m-t+d_1+\cdots+d_\nu$ rows have at most $d_1+\cdots+d_\nu+f_\nu$ nonzero entries in the first $d_1+\cdots+d_\nu+f_\nu$ columns and zeros afterwards.
\item its last $m-t+d_1+\cdots+d_{\nu-1}$ rows have at most $d_1+\cdots+d_\nu$ nonzero entries in the first $d_1+\cdots+d_\nu$ columns and zeros afterwards.
\item its last $m-t+e_1+\cdots+e_{\nu-1}$ columns have at most $e_1+\cdots+e_\nu$ nonzero entries in the first $e_1+\cdots+e_\nu$ rows and zeros afterwards.
\end{enumerate}

Let us define a square submatrix $B_2$ of $N$ of order $b_2$, with $b_2 < b_1$, and such that $B_2 \tilde u = 0$ where $\tilde u$ is a column matrix whose entries are elements of $u$.
From the inequality (\ref{m-t}), three cases may happen:

\begin{enumerate}
\item If $m-t\geq f_\nu$, let $b_2=d_1+\cdots+d_\nu+f_\nu$ and take $B_2$ to be a square submatrix of order $b_2$, of the matrix formed by the last $m-t+d_1+\cdots+d_\nu$ rows of $N$ and the first $d_1+\cdots+d_\nu+f_\nu$ columns of $N$.
\item If $m-t\geq d_\nu$, let $b_2=d_1+\cdots+d_\nu$ and take $B_2$ to be a square submatrix of order $b_2$, of the matrix formed by the last $m-t+d_1+\cdots+d_{\nu-1}$ rows of $N$ and the first $d_1+\cdots+d_\nu$ columns of $N$.
\item If $m-t\geq e_\nu$, let $b_2=t-(e_1+\cdots+e_{\nu-1})$ and take $B_2$ to be a square submatrix of order $b_2$, of the matrix formed by the last $m-(e_1+\cdots+e_{\nu})$ rows of $N$ and the first $t-(e_1+\cdots+e_{\nu-1})$ columns of $N$.
\end{enumerate}

In either case, choosing $\widetilde{u}=[u_{i_1}, \dots, u_{i_{b_2}}]^T$, accordingly, we have $B_2\widetilde{u}=0$. Notice that $b_2<t<m<b_1$. Since $N$ is superregular, $\mid B_2\mid$ is a trivial minor which implies that the columns of $B_2$ are linearly dependent. Also every row of $B_2$ will have at least two nonzero entries. Hence $B_2$ has the same properties as $B_1$.

Hence, using infinite descent, we always get a contradiction. Thus
\[\mathrm{wt}(Bu)\geq a-b+1.\]

\vspace{4mm}
To finalize the proof we will show that the assumptions (\ref{empty}) and (\ref{zeros2}) are satisfied.
\vspace{2mm}

i) Proof of assumption (\ref{empty}): let $1\leq \lambda \leq \nu$ and $k\in D_\lambda$. Then, by (\ref{defio}), there exists $i_{\lambda-1}\in D_{\lambda-1}$ such that $\mu_{i_{\lambda-1}\,k}\neq 0$. Let $j_{\lambda-1}\in E_{\lambda-1}$. We are going to prove that $\mu_{k\,j_{\lambda-1}}=0$ and, so, $k\notin E_\lambda$.

Since $i_{\lambda-1}\in D_{\lambda-1}$ then, by (\ref{defio}), there exist $i_0\in D_0$, $i_1\in D_1, \dots,$ $i_{\lambda-2}\in D_{\lambda-2}$, all different, such that $\mu_{i_\ell\,i_{\ell+1}}\neq 0$ for $0\leq\ell\leq\lambda-2$. Moreover, since $j_{\lambda-1}\in E_{\lambda-1}$ then, by (\ref{defio}), there exist $j_0\in E_0$, $j_1\in E_1, \dots,$ $j_{\lambda-2}\in E_{\lambda-2}$, all different, such that $\mu_{j_{\ell+1}\,j_\ell}\neq 0$ for $0\leq\ell\leq\lambda-2$.

Consider the matrix $V_{i_0\,j_0}$, defined in (\ref{vij}), and the permutation $\widetilde{\sigma}\in S_{t+1}$ defined below, depending on $\lambda$. For $\lambda=1$, the permutation is defined by
\begin{itemize}
\item $\widetilde{\sigma}(k)=t+1$,
\item $\widetilde{\sigma}(t+1)=k$,
\item $\widetilde{\sigma}(s)=s$ for $s\in\{1,2,\dots,t\}\setminus\{k\}$,
\end{itemize}
For $\lambda=2$, by
\begin{itemize}
\item $\widetilde{\sigma}(i_1)=k$,
\item $\widetilde{\sigma}(k)=j_1$,
\item $\widetilde{\sigma}(j_1)=t+1$,
\item $\widetilde{\sigma}(t+1)=i_1$,
\item $\widetilde{\sigma}(s)=s$ for $s\in\{1,2,\dots,t\}\setminus\{i_1,j_1,k\}$,
\end{itemize}
And, for $\lambda\geq 3$, by
\begin{itemize}
\item $\widetilde{\sigma}(i_{\lambda-1})=k$,
\item $\widetilde{\sigma}(k)=j_{\lambda-1}$,
\item $\widetilde{\sigma}(j_{\ell+1})=j_\ell$, for $1\leq\ell\leq\lambda-2$,
\item $\widetilde{\sigma}(j_1)=t+1$,
\item $\widetilde{\sigma}(t+1)=i_1$,
\item $\widetilde{\sigma}(i_{\ell})=i_{\ell+1}$, for $1\leq\ell\leq\lambda-2$,
\item $\widetilde{\sigma}(s)=s$ for $s\in\{1,2,\dots,t\}\setminus\{i_1,\dots,i_{\lambda-1},j_1,\dots,j_{\lambda-1},k\}$.
\end{itemize}

Now, using the superregularity of $\widehat{B}$, we conclude that $\mu_{k\,j_{\lambda-1}}=0$. Thus, $k\notin E_{\lambda}$.

Similarly, if $k\in E_{\lambda}$ then $\mu_{i k}=0$ for all $i\in D_{\lambda-1}$. Therefore, $k\notin D_\lambda$. Hence $D_\lambda\cap E_{\lambda}=\emptyset$.

\vspace{7mm}
ii) Proof of assumption (\ref{zeros2}): let $1\leq \lambda \leq \nu$, $i_\lambda\in D_\lambda$ and $j_\lambda\in E_\lambda$. Then, by (\ref{defio}), there exist sequences of integers $i_0\in D_0$, $i_1\in D_1, \dots,$ $i_{\lambda-1}\in D_{\lambda-1}$, all different, such that $\mu_{i_\ell\,i_{\ell+1}}\neq 0$ for $0\leq\ell\leq\lambda-1$, and $j_0\in E_0$, $j_1\in E_1, \dots,j_{\lambda-1}\in E_{\lambda-1}$, all different, such that $\mu_{j_{\ell+1}\,j_\ell}\neq 0$ for $0\leq\ell\leq\lambda-1$. Consider the matrix $V_{i_0,j_0}$ defined in (\ref{vij}) and the permutation $\widetilde{\sigma}\in S_{t+1}$ defined below.

If $\lambda=1$ then $\widetilde{\sigma}$ is defined by
\begin{itemize}
\item $\widetilde{\sigma}(i_1)=j_1$
\item $\widetilde{\sigma}(j_1)=t+1$,
\item $\widetilde{\sigma}(t+1)=i_1$,
\item $\widetilde{\sigma}(s)=s$ for $s\in\{1,2,\dots,t\}\setminus\{i_1,j_1\}$,
\end{itemize}
and, if $\lambda\geq 2$, by
\begin{itemize}
\item $\widetilde{\sigma}(i_{\lambda-1})=i_\lambda$,
\item $\widetilde{\sigma}(i_\lambda)=j_\lambda$
\item $\widetilde{\sigma}(j_{\ell+1})=j_\ell$, for $1\leq\ell\leq\lambda-1$,
\item $\widetilde{\sigma}(j_1)=t+1$,
\item $\widetilde{\sigma}(t+1)=i_1$,
\item $\widetilde{\sigma}(i_{\ell})=i_{\ell+1}$, for $1\leq\ell\leq\lambda-2$,
\item $\widetilde{\sigma}(s)=s$ for $s\in\{1,2,\dots,t\}\setminus\{i_1,\dots,i_\lambda,j_1,\dots,j_\lambda\}$.
\end{itemize}
Hence, we obtain $\mu_{i_\lambda\,j_\lambda}=0$. Therefore, (\ref{zeros2}) is valid.
\end{proof}

The following example illustrates the procedure described in the proof of the previous theorem.

\begin{example}
Suppose $a=11$, $b=10$ and $\mathbb F$ a finite field. In the matrices described below, $\times$ stands for a entry that is nonzero and $0$ for a entry that is zero. All the other entries may be zero or nonzero. Let
\[B=\left [ \begin{array}{cccccccccc}
\times & & & & & & & & 0 &  \\
& \times & & & & & & & 0 &  \\
 & & \times & & & & & & 0 & \times \\
 & & & \times & & & & & 0 & \\
 & &  & & \times & & & & 0 & \times \\
 & & & & & \times & & & 0 & \\
 & & & & & & \times & & 0 & \\
 & & & & & &  & \times & 0 & \\
 &  \times & & & &  \times & & & 0 & \\
\times & & & \times & &  & & & 0 & \\
\times & & & & & & & & \times & \times
\end{array} \right ] \in \mathbb F^{a \times b}
\]
be a superregular matrix and $u=[u_1,\dots u_{10}]^T$ such that $Bu=0$ with $u_i\neq 0$, for $1\leq i\leq 10$. So the columns of $B$ are linearly dependent. Suppose that $B_1$ is the submatrix of $B$ obtained by deleting the last row,
\[B_1=\left [ \begin{array}{cccccccccc}
\times & & & & & & & & 0 &  \\
& \times & & & & & & & 0 &  \\
 & & \times & & & & & & 0 & \times \\
 & & & \times & & & & & 0 & \\
 & &  & & \times & & & & 0 & \times \\
 & & & & & \times & & & 0 & \\
 & & & & & & \times & & 0 & \\
 & & & & & &  & \times & 0 & \\
 &  \times & & & &  \times & & & 0 & \\
\times & & & \times & &  & & & 0 &
\end{array} \right ].
\]
Since the next to last column is identically zero, all the other columns are linear dependent. So, we consider the matrices
\[\overline{B}=\left [ \begin{array}{ccccccccc}
\times & & & & & & & &   \\
& \times & & & & & & &   \\
 & & \times & & & & & &  \times \\
 & & & \times & & & & &  \\
 & &  & & \times & & & &  \times \\
 & & & & & \times & & &   \\
 & & & & & & \times & &   \\
 & & & & & &  & \times &  \\
 &  \times & & & &  \times & & &  \\
\times & & & \times & &  & & &
\end{array} \right ],\qquad\widehat{B}=\left [ \begin{array}{cccccccc|c}
\times & & & & & & & &  \\
& \times & & &  & & & &  \\
 & & \times & & & & & & \times \\
 & & & \times & & & & & \\
 & &  & & \times  & & & & \times \\
 & & & & & \times & &  & \\
 & & & & & & \times &  & \\
 & & & & & &  & \times  & \\ \hline
 &  \times & & & &  \times  & &  &
\end{array} \right ],
\]
where $\widehat{B}=[\mu_{ij}]$ is a square submatrix of $\bar B$ of order $m=9$, obtained form $\bar B$ by deleting its last row. Let us assume that $t = \mbox{rank} \, \widehat{B}=8$ and that $\widetilde{B}$ formed by the first $8$ rows and the first $8$ columns of $\widehat{B}$ is nonsingular. Since $\mid\widehat{B}\mid=0$ and $B$ is superregular, $\mid\widehat{B}\mid$ is a trivial minor, so using the permutation $\sigma(i)=i$, we get $\mu_{9\,9}=0$. With the permutations

\[\left (\begin{array}{ccccccccc}
1 & 2 & 3 & 4 & 5 & 6 & 7 & 8 & 9\\
1 & 9 & 3 & 4 & 5 & 6 & 7 & 8 & 2
\end{array}\right ),\; \left (\begin{array}{ccccccccc}
1 & 2 & 3 & 4 & 5 & 6 & 7 & 8 & 9\\
1 & 2 & 3 & 4 & 5 & 9 & 7 & 8 & 6
\end{array}\right )\]
and
\[\left (\begin{array}{ccccccccc}
1 & 2 & 3 & 4 & 5 & 6 & 7 & 8 & 9\\
1 & 2 & 9 & 4 & 5 & 6 & 7 & 8 & 3
\end{array}\right ),\; \left (\begin{array}{ccccccccc}
1 & 2 & 3 & 4 & 5 & 6 & 7 & 8 & 9\\
1 & 2 & 3 & 4 & 9 & 6 & 7 & 8 & 5
\end{array}\right )\]
we obtain $\mu_{2\,9}=\mu_{6\,9}=\mu_{9\,3}=\mu_{9\,5}=0$. Hence
\[ M=\widehat{B}=\left [ \begin{array}{cccccccc|c}
\times & & & & & & & &  \\
& \times & & &  & & & & 0  \\
 & & \times & & & & & & \times \\
 & & & \times & & & & & \\
 & &  & & \times  & & & & \times \\
 & & & & & \times & &  & 0 \\
 & & & & & & \times &  & \\
 & & & & & &  & \times  & \\ \hline
 &  \times & 0 & & 0 &  \times  & &  & 0
\end{array} \right ].
\]
Assume that all the other entries of the last row and all the other entries of the last column which are not represented in $M$ are zero. Then $D_1=\{2,6\}$ and $d_1=2$, $E_1=\{3,5\}$ and $e_1=2$ and $F_1=\{1,4,7,8\}$ and $f_1=4$. Now consider the pairs $(2,3),(2,5),(6,3)$ and $(6,5)$. The permutations $\widetilde{\sigma}$ defined by
\[\left (\begin{array}{ccccccccc}
1 & 2 & 3 & 4 & 5 & 6 & 7 & 8 & 9\\
1 & 3 & 9 & 4 & 5 & 6 & 7 & 8 & 2
\end{array}\right ),\; \left (\begin{array}{ccccccccc}
1 & 2 & 3 & 4 & 5 & 6 & 7 & 8 & 9\\
1 & 5 & 3 & 4 & 9 & 9 & 7 & 8 & 2
\end{array}\right )\]
and
\[\left (\begin{array}{ccccccccc}
1 & 2 & 3 & 4 & 5 & 6 & 7 & 8 & 9\\
1 & 2 & 9 & 4 & 5 & 3 & 7 & 8 & 6
\end{array}\right ),\; \left (\begin{array}{ccccccccc}
1 & 2 & 3 & 4 & 5 & 6 & 7 & 8 & 9\\
1 & 2 & 3 & 4 & 9 & 5 & 7 & 8 & 6
\end{array}\right )\]
enable us to conclude that $\mu_{2\,3}=\mu_{2\,5}=\mu_{6\,3}=\mu_{6\,5}=0$ (see (\ref{zeros2})). So
\[M=\left [ \begin{array}{cccccccc|c}
\times & & & & & & & & 0 \\
& \times & 0 & & 0 & & & & 0  \\
 & & \times & & & & & & \times \\
 & & & \times & & & & & 0 \\
 & &  & & \times  & & & & \times \\
 & & 0 & & 0 & \times & &  & 0 \\
 & & & & & & \times &  & 0 \\
 & & & & & &  & \times  & 0 \\ \hline
0 &  \times & 0 & 0 & 0 &  \times  & 0 & 0 & 0
\end{array} \right ].
\]
Suppose $\mu_{2\,1}\neq 0$, $\mu_{6,4}\neq 0$, $\mu_{7,3}\neq 0$ and $\mu_{8,5}\neq 0$, then $D_2=\{1,4\}$, $E_2=\{7,8\}$ and $F_2=\emptyset$. Also, $d_2=2$, $e_2=2$, $f_2=0$. Moreover, from (\ref{zerosR1}) and (\ref{zerosC1}), we have that $\mu_{ij}=0$ for $(i,j)\in\{(1,3),(1,5),(4,3),(4,5),(2,7),(6,7),(2,8),(6,8)\}$, and
\[M=\left [ \begin{array}{cccccccc|c}
\times & & 0 & & 0 & & & & 0 \\
\times & \times & 0 & & 0 & & 0 & 0 & 0  \\
 & & \times & & & & & & \times \\
 & & 0 & \times & 0 & & & & 0 \\
 & &  & & \times  & & & & \times \\
 & & 0 & \times & 0 & \times & 0 & 0 & 0 \\
 & & \times & & & & \times &  & 0 \\
 & & & & \times & &  & \times  & 0 \\ \hline
0 &  \times & 0 & 0 & 0 &  \times  & 0 & 0 & 0
\end{array} \right ].
\]

Now, we if we use the following permutations

\[\left (\begin{array}{ccccccccc}
1 & 2 & 3 & 4 & 5 & 6 & 7 & 8 & 9\\
7 & 1 & 9 & 4 & 5 & 6 & 3 & 8 & 2
\end{array}\right ),\;
\left (\begin{array}{ccccccccc}
1 & 2 & 3 & 4 & 5 & 6 & 7 & 8 & 9\\
8 & 1 & 3 & 4 & 9 & 6 & 7 & 5 & 2
\end{array}\right )\]
and
\[\left (\begin{array}{ccccccccc}
1 & 2 & 3 & 4 & 5 & 6 & 7 & 8 & 9\\
1 & 2 & 9 & 7 & 5 & 4 & 3 & 8 & 6
\end{array}\right ),\;
\left (\begin{array}{ccccccccc}
1 & 2 & 3 & 4 & 5 & 6 & 7 & 8 & 9\\
1 & 2 & 3 & 8 & 9 & 4 & 7 & 5 & 6
\end{array}\right ),\]
we obtain $\mu_{ij}=0$ for $(i,j)\in\{(1,7), (1,8), (4,7), (4,8)\}$.

Therefore,
\[M=\left [ \begin{array}{cccccccc|c}
\times & & 0 & & 0 & & 0 & 0 & 0 \\
\times & \times & 0 & & 0 & & 0 & 0 & 0  \\
 & & \times & & & & & & \times \\
 & & 0 & \times & 0 & & 0 & 0 & 0 \\
 & &  & & \times  & & & & \times \\
 & & 0 & \times & 0 & \times & 0 & 0 & 0 \\
 & & \times & & & & \times &  & 0 \\
 & & & & \times & &  & \times  & 0 \\ \hline
0 &  \times & 0 & 0 & 0 &  \times  & 0 & 0 & 0
\end{array} \right ].
\]

Before proceed, we will perform permutations on the rows and columns of $M$ so that the zeros are moved to the right bottom corner. By making first a permutation of the rows and then a permutation of the columns, we obtain
\[\mid M \mid=\pm\left | \begin{array}{cccccccc|c}
 & & \times & &  & & & & \times \\
 &  &  & & \times & &  &  & \times  \\
 & & \times &  &  & & \times & & 0 \\
 & &  & & \times  & & & \times & 0 \\
\times & & 0 &  & 0 &  & 0 & 0 & 0 \\
 & & 0 & \times & 0 & & 0 & 0 & 0 \\
\times & \times & 0 & & 0 & & 0 & 0 & 0 \\
 & & 0 & \times & 0 & \times & 0 & 0  & 0 \\ \hline
0 &  \times & 0 & 0 & 0 &  \times  & 0 & 0 & 0
\end{array} \right |\qquad =\quad \pm\left | \begin{array}{cccccccc|c}
 & &  & &  & & & \times & \times \\
 &  &  & &  & & \times &  & \times  \\
 & &  &  &  & \times &  & \times & 0 \\
 & &  & & \times  & &  \times &  & 0 \\
 & & \times &  & 0 & 0 & 0  & 0 & 0 \\
 & &  & \times & 0 & 0 & 0  & 0 & 0 \\
\times & & \times & & 0 & 0  & 0 & 0 & 0 \\
 & \times &  & \times & 0 & 0 & 0 & 0  & 0 \\ \hline
\times &  \times & 0 & 0 & 0 &  0  & 0 & 0 & 0
\end{array} \right |.
\]

Since $m-t > f_2$, we consider $B_2$ equal to the matrix formed by the rows $5, 6, 7$ and $8$, and the columns $1, 2, 3$ and $4$ of the last matrix, i. e.
\[B_2=\left [\begin{array}{cccc}
& & \times &  \\
 & &  & \times \\
 \times & & \times \\
  & \times &  & \times
 \end{array} \right ].
\]
With $\widetilde{u}$ appropriately chosen we have $B_2\widetilde{u}=0$ and so $\mid B_2\mid =0$. But the term corresponding to the permutation $\sigma(1)=3$, $\sigma(2)=4$, $\sigma(3)=1$ and $\sigma(4)=2$ is nontrivial. Hence we have one nontrivial minor equal to zero, contradicting the hypothesis that $B$ is superregular.

Therefore, \[\mathrm{wt}(Bu)\geq 11-10+1\geq 2.\]

\end{example}

The next theorem states that matrices over $\mathbb{F}$ of a certain form are superregular. Similar matrices were defined in \cite{ANP2013}.

\begin{theorem}\label{diagonal} Let $\alpha$ be a primitive element of a finite field $\mathbb{F}=\mathbb{F}_{p^N}$ and $B=[\nu_{i\,\ell}]$ be a matrix over $\mathbb{F}$ with the following properties
\begin{enumerate}
\item if $\nu_{i\,\ell}\neq 0$ then $\nu_{i\,\ell}=\alpha^{\beta_{i\,\ell}}$ for a positive integer $\beta_{i\,\ell}$;
\item If $\nu_{i\,\ell}=0$ then $\nu_{i'\,\ell}=0$, for any $i'>i$ or  $\nu_{i\,\ell'}=0$, for any $\ell'<\ell$;
\item if $\ell<\ell'$, $\nu_{i\ell}\neq 0$ and $\nu_{i\ell'}\neq 0$ then $2\beta_{i\,\ell}\leq\beta_{i\,\ell'}$;
\item if $i<i'$, $\nu_{i\,\ell}\neq 0$ and $\nu_{i'\,\ell}\neq 0$ then $2\beta_{i\,\ell}\leq\beta_{i'\,\ell}$.
\end{enumerate}
Suppose $N$ is greater than any exponent of $\alpha$ appearing as a nontrivial term of any minor of $B$. Then $B$ is superregular.
\end{theorem}
\begin{proof}
Let $C=[c_{a\,b}]$ be a square submatrix of $B$ of order $m$ such that $\mid C\mid $ is a nontrivial minor. We are going to prove that $\mid C\mid \neq 0$.

Let $C_1, \dots, C_m$ be the columns of $C$. Firstly, we will define, recursively, a sequence of integers $i_1, i_2, \dots, i_m$, such that the antidiagonal term of the minor $\left |C_{i_1}\, C_{i_2}\, \dots\, C_{i_m}\right |$ is nontrivial.

Since $\mid C\mid $ has a nontrivial term, the last row of $C$ must have a nonzero entry. Define $i_1=\min \{i\,|\, c_{m\,i}\neq 0\}$. Given $j\in\{2,3,\dots,m-1\}$, suppose $i_1, i_2, \dots, i_{j-1}$ are well defined and take the set
\[I_j=\{i\,|\, c_{m-j+1\,i}\neq 0\,\mbox{and}\, i\notin\{i_k\,|\,k<j\}\}.\]

Suppose that $I_j=\emptyset$ then $c_{m-j+1\,i}=0$ for any $i\notin\{i_k\,|\,k<j\}$. Let $\sigma\in S_m$ be a permutation such that $c_\sigma$ is a nontrivial term of $\mid C\mid$. Clearly, $\sigma(m-j+1)=i_{k_1}$, for some $k_1\in\{1,2, \dots,j-1\}$. Let $\ell_1=\sigma(m-k_1+1)$. Then $\ell_1\neq i_{k_1}$. Suppose $\ell_1> i_{k_1}$. If $c_{m-j+1\,\ell_1}=0$ then, by propriety $2.$, $c_{m-j+1\, i_{k_1}}=0$ or $c_{m-k_1+1\,\ell_1}=0$ contradicting the fact that $c_\sigma$ is a nontrivial term. Therefore, $c_{m-j+1\,\ell_1}\neq 0$ and so $\ell_1\in\{i_k\,|\,k<j\}\setminus\{i_{k_1}\}$. If $\ell_1< i_{k_1}$ then, by definition of $i_{k_1}$, $\ell_1\in\{i_k\,|\,k<j\}\setminus\{i_{k_1}\}$. Now, for $r\in\{2, 3, \dots, j-1\}$, and using a similar reasoning, we may take $k_r$ such that $i_{k_r}=\ell_{r-1}$ and $\ell_r=\sigma(m-k_r+1)$. But then
 \[\ell_{j-1}\in\{i_k\,|\,k<j\}\setminus\{i_{k_1},i_{k_2},\dots, i_{k_{j-1}}\}=\emptyset,\]
 which is impossible. Hence $I_j\neq\emptyset$, and so we may define
 \[i_j=\min\{i\,|\, c_{m-j+1\,i}\neq 0\,\mbox{and}\, i\notin\{i_k\,|\,k<j\}\}.\]
 Thus, the integers $i_1, i_2, i_m$ are well defined. Notice that if the antidiagonal term of $\mid C\mid $ is nontrivial then, clearly, $i_j=j$, for $j\in\{1,3,\dots,m\}$.

Now, define $A=\left [C_{i_1}\, C_{i_2}\, \dots\, C_{i_m}\right ]=[\mu_{i\,\ell}]$. Clearly, the matrix $A$ satisfies propriety $1.$ but also the following proprieties
\begin{description}
\item[(i)] if $\hat{\sigma}\in S_m$ is the permutation defined by $\hat{\sigma}(i)=m-i+1$, then $\mu_{\hat{\sigma}}$ is a nontrivial term of $\mid A\mid $.
\item[(ii)] if $\ell\geq m-i+1$, $\ell<\ell'$, $\mu_{i\ell}\neq 0$ and $\mu_{i\ell'}\neq 0$ then $2\beta_{i\ell}\leq\beta_{i\ell'}$;
\item[(iii)] if $\ell\geq m-i+1$, $i<i'$, $\mu_{i\ell}\neq 0$ and $\mu_{i'\ell}\neq 0$ then $2\beta_{i\ell}\leq\beta_{i'\ell}$.
\end{description}

Let $\sigma\in S_m$ such that $\mu_{\sigma}$ is a nontrivial term of $\mid A\mid $. By property $1.$, we have $\mu_\sigma=\alpha^{\beta_\sigma}$, for a positive integer $\beta_\sigma$.

Let $T_m=\{\sigma\in S_m\; |\; \sigma\neq \hat{\sigma}\;\mbox{and}\; \mu_\sigma\; \mbox{is a nontrivial term of} \;\mid A\mid \}$. If $T_m=\emptyset$ then
$\mid A\mid  =\mu_{\hat{\sigma}}=\alpha^{\beta_{\hat{\sigma}}} \neq 0$.

If $T_m\neq\emptyset$, let $\sigma\in T_m$. We are going to prove that $\beta_{\hat{\sigma}}<\beta_\sigma$. Since $\mu_\sigma$ in a nontrivial term of $\mid A\mid $, for any $1\leq i\leq m$, there exists $\ell\geq i$ such that $\sigma(\ell)\geq m-i+1$. For any $1\leq \ell\leq m$ define

\[S_\ell=\{i\; |\; i\leq \ell\; \mbox{and}\; \sigma(\ell)\geq m-i+1\}.\]

Notice that $\displaystyle{\cup_{1\leq j\leq m}S_\ell=\{1,2, \dots, m\}}$ and, since $\sigma\neq \hat{\sigma}$, there exists at least one $\ell_0$, such that $1\leq \ell_0\leq m$ and $S_{\ell_0}=\emptyset$. By properties (ii) and (iii), we have that if $S_\ell\neq\emptyset$,
\[\sum_{i\in S_\ell}\beta_{i\; m-i+1} \leq \beta_{\ell\;\sigma(\ell)}.\]

Therefore
\[\beta_{\hat{\sigma}} = \sum_{i=1}^m\beta_{i\; m-i+1} \leq \sum_{\stackrel{\ell=1}{S_\ell\neq\emptyset}}^m\beta_{\ell\;\sigma(\ell)} < \sum_{\ell=1}^m \beta_{\ell\;\sigma(\ell)}.\]

So
\[\mid A\mid =\alpha^{\beta_{\hat{\sigma}}} +\sum_{h=\beta_{\hat{\sigma}}+1}^{N-1}\epsilon_h\alpha^h,\]
 where $\epsilon_h\in \{0,1,\dots, p-1\}$. Hence $\mid A\mid \neq 0$ and so $\mid C\mid \neq 0$.

 Thus $B$ is superregular.
\end{proof}

The following examples illustrates the procedure described in the proof of the previous theorem.

\begin{example}
Let $E=[e_{ij}]$ be the matrix
\[\left[\begin{array}{cccccc}
\emptyset         & \emptyset & 2 & 3 & 4 & 5\\
0 & 1 & 3 & 4 & 5 & 6\\
1         & 2 & 4 & 5 & 6 & 7\\
2         & \emptyset & 5 & 6 & 7 & 8\\
\emptyset         & \emptyset & 6 & 7 & \emptyset & 9\\
\emptyset         & \emptyset & 7 & 8 & \emptyset & \emptyset
\end{array}
\right]\]
and $C=[c_{ij}]$ be the $6 \times 6$ matrix defined by
 \[c_{ij}=
 \left\{
 \begin{array}{ll}
 0 & \mbox{if} \; e_{ij}= \emptyset \\
 \alpha^{2^{e_{ij}}} & \mbox{elsewhere}
 \end{array}
 \right. .\]
In this case, $i_1=3$, $i_2=4$, $i_3=1$, $i_4=2$, $i_5=5$ and $i_6=6$. Therefore, the matrix $F=[f_{ij}]=[E_3\, E_4\, E_1\, E_2\, E_5\, E_6]$, where $E_i$ represents the $i$-$th$ column of $E$, is
\[\left[\begin{array}{cccccc}
2 & 3 & \emptyset         & \emptyset & 4 & 5\\
3 & 4 & 0 & 1 & 5 & 6\\
4 & 5 & 1         & 2 & 6 & 7\\
5 & 6 & 2         & \emptyset & 7 & 8\\
6 & 7 & \emptyset         & \emptyset & \emptyset & 9\\
7 & 8 & \emptyset         & \emptyset & \emptyset & \emptyset
\end{array}
\right]\]
and therefore $A=[a_{ij}]$, the $6 \times 6$ matrix, defined by
 \[a_{ij}=
 \left\{
 \begin{array}{ll}
 0 & \mbox{if} \; f_{ij}= \emptyset \\
 \alpha^{2^{f_{ij}}} & \mbox{elsewhere}
 \end{array}
 \right. \]
satisfies proprieties (i), (ii) and (iii).
\end{example}

\begin{example}
Consider the matrix
\[A=
\left[\begin{array}{ccccccc}
0 & 0 & 0 & \alpha^{2^3} & \alpha^{2^4} & \alpha^{2^{12}} & \alpha^{2^{13}}\\
0 & 0 & 0 & \alpha^{2^6} & \alpha^{2^7} & \alpha^{2^{15}} & \alpha^{2^{16}}\\
0 &\alpha^{2^0} & \alpha^{2^1} & \alpha^{2^9} & \alpha^{2^{10}} &\alpha^{2^{18}} & 0 \\
0 &\alpha^{2^3} & \alpha^{2^4} & \alpha^{2^{12}} & \alpha^{2^{13}} &\alpha^{2^{21}} & 0 \\
0 &\alpha^{2^6} & \alpha^{2^7} & \alpha^{2^{15}} & \alpha^{2^{16}} &\alpha^{2^{24}} & 0 \\
\alpha^{2^1} & \alpha^{2^9} & \alpha^{2^{10}} &\alpha^{2^{18}} & 0 & 0 & 0 \\
\alpha^{2^4} & \alpha^{2^{12}} & \alpha^{2^{13}} &\alpha^{2^{21}} & 0 & 0 & 0
\end{array}\right].\]
Let $\hat{\sigma}\in S_7$ be the permutation defined by $\hat{\sigma}(i)=8-i$ and let
\[\sigma=\left(\begin{array}{ccccccc}
1 & 2 & 3 & 4 & 5 & 6 & 7\\
6 & 7 & 5 & 3 & 2 & 4 & 1
\end{array}\right).\]
Clearly $\mu_{\sigma}$ is a nontrivial term of $\mid A\mid $. The next table shows the sets $S_\ell$, for $1\leq\ell\leq m$.

\vspace{7mm}
\begin{center}
\begin{tabular}{|c|c|c|c|c|c|c|c|}
\hline
$\ell\quad$ & $1$    & $2$    & $3$    & $4$    & $5$    & $6$    & $7$   \\ \hline
$S_\ell$  & $\emptyset\qquad$ &  $\{1,2\}$  & $\{3\}\quad$ & $\emptyset\qquad$   & $\emptyset\qquad$   & $\{4,5,6\}$   & $\{7\}\quad$\\ \hline
\end{tabular}
\end{center}
\vspace{7mm}
Now,
\begin{eqnarray*}
\sum_{i\in S_2}\beta_{i\; m-i+1} & = & \beta_{1\; 7}+ \beta_{2\; 6}\\
& = & 2^{13}+2^{15}\\
& < & 2^{16}\\
& = & \beta_{2\;\sigma(2)},
\end{eqnarray*}
for $\ell=3$ and $\ell=7$ we have
\begin{eqnarray*}
\sum_{i\in S_3}\beta_{i\; m-i+1} & =  \beta_{3\; 5} & = \beta_{3\;\sigma(3)}\\
\sum_{i\in S_7}\beta_{i\; m-i+1} & =  \beta_{7\; 1} & = \beta_{7\;\sigma(7)}
\end{eqnarray*}
and for $\ell=6$ we have
\begin{eqnarray*}
\sum_{i\in S_6}\beta_{i\; m-i+1} & = & \beta_{4\; 4}+ \beta_{5\; 3}+\beta_{6\; 2}\\
& = & 2^{12}+2^{7}+2^{9}\\
& < & 2^{18}\\
& = & \beta_{6\;\sigma(6)}.
\end{eqnarray*}
So
\[\sum_{i=1}^m\beta_{i\; m-i+1}<\sum_{\ell=1}^m \beta_{\ell\;\sigma(\ell)}.\]
Note that any nontrivial term of $A$ has determinant smaller than $\alpha^{2^{25}}$. Then, if $N\geq 2^{25}$ we have $\mid A\mid \neq 0$.
\end{example}

\section{Constructions of optimal convolutional codes}\label{sec:construction}



Let $\cal C$ be a convolutional code of rate $k/n$ and different Forney $\nu_1 <  \dots < \nu_{\ell}$ with corresponding multiplicities $m_1, \dots, m_{\ell}$ and
\[G(z)= \sum_{i=0}^{\nu_\ell} G_i z^i\]
a column reduced encoder of $\cal C$ with column degrees in nondecreasing order. Consider a nonzero codeword $v(z)=G(z)u(z)$ with $u(z)\in \mathbb F[z]^k$. Writing
\[ u(z)=\sum_{i=0}^\epsilon u_i z^i\quad\mbox{and}\quad v(z)=\sum_{i=0}^{\nu_\ell+\epsilon} v_i z^i,\]
we have

\[
\left[\begin{array}{c}
v_0\\
v_1\\
v_2\\
\vdots\\
v_{\nu_\ell+\epsilon}
\end{array}
\right] ={\cal G}(\epsilon)\quad \left[\begin{array}{c}
u_\epsilon\\
\vdots \\
u_1\\
u_0
\end{array}
\right]
\]

where
\begin{equation} \label{Gd}
{\cal G}(\epsilon)=\left[
\begin{array}{cccccc}
0 & 0 & \cdots & 0 & 0 & G_0 \\
0 & 0 & \cdots & 0 & G_0 & G_1 \\
\vdots & \vdots & \ddots & \vdots & \vdots & \vdots \\
0 & 0 & \cdots & G_{\nu_\ell-2} & G_{\nu_\ell-1} & G_{\nu_\ell}\\
0 & 0 & \cdots & G_{\nu_\ell-1} & G_{\nu_\ell} &  0 \\
0 & 0 & \cdots & G_{\nu_\ell} & 0 &  0 \\
\vdots & \vdots & \ddots &\vdots & \vdots & \vdots \\
 G_0 & G_1 & \cdots & 0 & 0 & 0 \\
 G_1 & G_2 & \cdots & 0 & 0 & 0 \\
\vdots & \vdots & \ddots &\vdots & \vdots & \vdots \\
G_{\nu_\ell-1} & G_{\nu_\ell}  & \cdots & 0 & 0 & 0 \\
G_{\nu_\ell} &  0 & \cdots & 0 & 0 & 0
 \end{array}
\right]\in \mathbb F^{n(\nu_\ell+\epsilon+1)\times k(\epsilon+1)}.
\end{equation}

We will prove that if $G(z)$ is such that the matrices ${\cal G}(\epsilon)$ defined in (\ref{Gd}) are superregular, for certain values of $\epsilon$,  then $\cal C$ is an optimal $(n,k,\nu_1,m_1)$ convolutional code.

\begin{theorem}\label{MDSCC}
Let $G(z)= \sum_{i \geq 0} G_i z^i \in \mathbb F[z]^{n \times k}$ be a matrix with column degrees $\nu_1<\cdots<\nu_\ell$ with multiplicities $m_1,\dots,m_\ell$, respectively, and such that all entries of the last $m_j + \cdots + m_{\ell}$  columns of $G_i$ are nonzero for $i \leq \nu_j$, $j=1, \dots, \ell$. Suppose that ${\cal G}(\epsilon_0)$, defined in (\ref{Gd}), is superregular for
\begin{equation}\label{epsilon0}
\epsilon_0=\left\lceil\frac{n(\nu_1+1)-m_1}{n-k}\right\rceil-1.
\end{equation}
Then $G(z)$ is column reduced and ${\cal C} = \mbox{Im}_{\mathbb F[z]}G(z)$ is an optimal $(n,k,\nu_1,m_1)$ convolutional code.

\end{theorem}

\begin{proof}
Clearly $k=m_1+\cdots+m_\ell$ and ${\cal G}(\epsilon)$ is superregular for any $\epsilon\leq\epsilon_0$. To prove that ${\cal C}$ is optimal we have to show that all nonzero codewords of ${\cal C}$, $v(z)$, have weight greater or equal than
 \[n(\nu_1 + 1) - m_1 +1.\]

Let $j\in\{1,2,\dots, \ell-1\}$. Note that if $i>\nu_j$, the first $m_1+\cdots+m_j$ columns of $G_i$ are zero and all the entries of the other columns of $G_i$ are nonzero.

Let $v(z)$ be a nonzero codeword of ${\cal C}$ and $u(z) \in \mathbb{F}[z]^k$ such that $v(z)=G(z)u(z)$. It is obvious that $u(z) \neq 0$.

\vspace{7mm}
Let us assume that $\epsilon\leq\epsilon_0$.

Suppose that the weight of $u(z)$ is $t$ and that deg $(u) \leq \epsilon$. Let $B$ be the matrix formed by the $t$ columns of ${\cal G}(\epsilon)$ that are multiplied by the nonzero entries of $u(z)$ to obtain $v(z)$. Next, we are going to calculate a lower bound for the number of rows of $B$ with nonzero entries depending on $t$.

If $(a-1)m_1<t\leq am_1$, for some $1\leq a\leq\nu_2-\nu_1$, then $B$ has at least $n(\nu_1+1+a-1)$ rows with nonzero entries. Since $B$ is superregular, using theorem \ref{SRZeros}, we obtain
\begin{eqnarray*}
\mathrm{wt}(v(z)) & \geq & n(\nu_1+a)-t+1\\
& = & n(\nu_1+1)+n(a-1)-t+1\\
& \geq & n(\nu_1+1)-m_1+1
\end{eqnarray*}
since $t-n(a-1) \leq t-m_1(a-1)\leq m_1$.

Let $b\in\{2, \dots, \ell\}$. For any $a$ such that
\[\left\{\begin{array}{ll}
1\leq a\leq\nu_{b+1}-\nu_b & \quad\mbox{if}\quad b< \ell.\\
1\leq a\leq\epsilon+ 1 & \quad\mbox{if}\quad b=\ell.
\end{array}\right.\]
and for $1\leq i\leq b$, define $\lambda_i=\min\{a-1+\nu_b-\nu_i,\epsilon+1\}$ and $\gamma_i=\min\{a+\nu_b-\nu_i,\epsilon+1\}$. Suppose that
\begin{equation}\label{columns}
\sum_{i=1}^b \lambda_im_i<t\leq\sum_{i=1}^b \gamma_im_i,
\end{equation}
then $B$ has at least $\,n(\nu_b+1+a-1)$ rows with nonzero entries. Again, using theorem \ref{SRZeros}, we obtain
\[\mathrm{wt}(v(z)) \geq  n(\nu_b+a) -t + 1.\]
On the other hand, we have that
\begin{eqnarray*}
t & \leq & \displaystyle{\sum_{i=1}^b(a+\nu_b-\nu_i)m_i} \\
& \leq & (a+\nu_b-\nu_1-1)m_1+m_1+\displaystyle{\sum_{i=2}^b(a+\nu_b-(\nu_1+1))m_i} \\
& \leq & m_1+n(a+\nu_b-(\nu_1+1))
\end{eqnarray*}
since $n>k\geq m_1+\cdots +m_b$. Hence, $\mathrm{wt}(v(z))\geq n(\nu_1+1)-m_1+1$ for every nonzero codeword.

\vspace{7mm}
Next suppose $\epsilon>\epsilon_0$. Let
\[v^{\epsilon_0}(z)=\sum_{i=0}^{\epsilon_0} v_i z^i\quad\mbox{and}\quad u^{\epsilon_0}(z)=\sum_{i=0}^{\epsilon_0} u_i z^i.\]
Note that the submatrix formed by the first $(\epsilon_0+1)n$ rows and the first $k(\epsilon-\epsilon_0)$ columns of ${\cal G}(\epsilon)$ is null.
Let $A$ be the matrix formed by the first $(\epsilon_0+1)n$ rows and the last $k(\epsilon_0+1)$ columns of ${\cal G}(\epsilon)$. Since $A$ is a submatrix of ${\cal G}(\epsilon_0)$, $A$ is superregular. The matrix $A$ is of the form
\[A=\left[
\begin{array}{ccccccccc}
0 & 0 & \cdots & 0 & 0 & \cdots & 0 & 0 & G_0 \\
0 & 0 & \cdots & 0 & 0 & \cdots & 0 & G_0 & G_1 \\
\vdots & \vdots & \ddots & \vdots & \vdots & \ddots & \vdots & \vdots & \vdots \\
0 & 0 & \cdots & 0 & 0 & \cdots & G_{\nu_\ell-2} & G_{\nu_\ell-1} & G_{\nu_\ell}\\
0 & 0 & \cdots & 0 & 0 & \cdots & G_{\nu_\ell-1} & G_{\nu_\ell} &  0 \\
0 & 0 & \cdots & 0 & 0 & \cdots & G_{\nu_\ell} & 0 &  0 \\
\vdots & \vdots & \ddots & \vdots & \vdots & \ddots & \vdots & \vdots & \vdots \\
0 & 0 & \cdots & 0 & G_0 & \cdots & 0 & 0 & 0\\
0 & 0 & \cdots & G_0 & G_1 & \cdots & 0 & 0 & 0\\
\vdots & \vdots & \ddots & \vdots & \vdots & \ddots & \vdots & \vdots & \vdots \\
0 & G_0 & \cdots & G_{\nu_\ell-2} &  G_{\nu_\ell-1}  & \cdots & 0 & 0 & 0 \\
 G_0 & G_1 & \cdots & G_{\nu_\ell-1} & G_{\nu_\ell} & \cdots & 0 & 0 & 0
 \end{array}
\right]
\]
Suppose that the weight of $u^{\epsilon_0}(z)$ is $t$. Let $B$ be the matrix formed by the $t$ columns of $A$ that are multiplied by the nonzero entries of $u^{\epsilon_0}(z)$ to obtain $v^{\epsilon_0}(z)$.

If all of the $n(\epsilon_0+1)$ rows of $B$ are nonzero, since $B$ has at most $k(\epsilon_0+1)$ nonzero columns and $B$ is superregular, then using theorem \ref{SRZeros} and (\ref{epsilon0}), we have
\begin{eqnarray*}
\mathrm{wt}(v^{\epsilon_0}(z)) & \geq & (n-k)(\epsilon_0+1)+1\\
& \geq &  n(\nu_1+1)-m_1+1.
\end{eqnarray*}

Now, suppose that $B$ has rows with all entries equal to zero (the number of such rows is always a multiple of $n$ by the structure of the matrix ${\cal G}(\epsilon)$). Since we may assume without loss of generality that $u_0$ has nonzero entries, the first $n(\nu_1+1)$ rows of $B$ are nonzero.Let $c$ be the largest integer such that the first $cn$ rows of $B$ are nonzero. Notice that $c=\nu_b+a$, for some $b\in\{1,\dots,\ell\}$ and $a$ such that
\[\left\{\begin{array}{ll}
1\leq a\leq\nu_{b+1}-\nu_b & \quad\mbox{if}\quad b< \ell.\\
1\leq a\leq\epsilon_0-\nu_\ell+ 1 & \quad\mbox{if}\quad b=\ell.
\end{array}\right.\]

With a similar argument as the one we used in the case $\epsilon\leq\epsilon_0$, we may conclude that the number of columns of $B$ is at most $\gamma_1m_1+\cdots+\gamma_\ell m_\ell$, where $\gamma_i=a+\nu_b-\nu_i$, for $1\leq i\leq \ell$. Let $B'$ be the matrix formed by the first $n(\nu_b+a)$ rows of $B$. Using the superregularity of $B'$ and theorem \ref{SRZeros}, we obtain
\begin{eqnarray*}
\mathrm{wt}(v^{\epsilon_0}(z)) & \geq & n(\nu_b+a)-\displaystyle{\sum_{i=1}^b(a+\nu_b-\nu_i)m_i}+1\\
& \geq & n(\nu_1+1)-m_1+1.
\end{eqnarray*}

Finally we prove that $\mathcal{C}$ has Forney indices $\nu_1, \nu_2, \cdots, \nu_{\ell}$ with multiplicities $m_1, m_2, \dots, m_{\ell}$, respectively. For that, it is sufficient to prove that $G(z)$ is column reduced, i.e., that $G^{hc}$ is full column rank. Notice that $G^{hc}$ is a submatrix of $G(\nu_{\ell} - \nu_1)$ constituted by nonzero entries, which means that all its $k \times k$ minors are different from zero. Consequently, $G^{hc}$ is full column rank and $G(z)$ is column reduced.

Therefore, the convolutional code ${\cal C} = \mbox{im}_{\mathbb F[z]}G(z)$ is an optimal $(n,k,\nu_1,m_1)$ convolutional code.
\end{proof}

Given any $n$ and $k$ with $n>k$, any $0\leq \nu_1< \dots< \nu_{\ell}$ and $m_1, \dots, m_{\ell}$ such that $k=m_1+\cdots+m_\ell$, we are going to construct an optimal $(n,k,\nu_1,m_1)$ convolutional code of rate $k/n$ over a finite field $\mathbb{F}=\mathbb{F}_{p^N}$, for $p$ prime and $N$ depending on $n$, $\nu_\ell$ and $\epsilon_0$ defined in (\ref{epsilon0}), with Forney indices $\nu_1,\dots,\nu_{\ell}$ and corresponding multiplicities $m_1, \dots, m_{\ell}$.

For $1\leq j\leq \ell-1$ and $0\leq i\leq \nu_\ell$, define $G_i \in \mathbb{F}^{n\times k}$ by
\begin{equation}\label{Gi}
G_i=[\gamma_{r\,s}(i)]\;\mbox{for}\;
\gamma_{r\,s}(i)=\left\{\begin{array}{lll}
\alpha^{2^{ni+r+s-2}} & \mbox{if} & i\leq\nu_1 \\
\alpha^{2^{ni+r+s-2}} & \mbox{if} & s>\displaystyle{\sum_{\kappa=1}^j m_\kappa}\quad\mbox{and}\quad\nu_j< i\leq \nu_{j+1} \\
0 & \mbox{if} & s\leq\displaystyle{\sum_{\kappa=1}^j m_\kappa}\quad\mbox{and}\quad\nu_j< i\leq \nu_{j+1}
\end{array}\right.
\end{equation}
where $\alpha$ is a primitive element of the finite field $\mathbb{F}$. If $N$ is greater than any exponent of $\alpha$ appearing as a nontrivial term of any minor of ${\cal G}(\epsilon_0)$ then ${\cal G}(\epsilon_0)$ satisfy the conditions of theorem \ref{diagonal} and so it is superregular. Using theorem \ref{MDSCC} we obtain the following result.

\begin{corollary}
Let $n, k, \ell \in \mathbb N$ such that $\ell \leq k < n$ and $\nu_1, \dots, \nu_{\ell}$, $m_1, \dots, m_{\ell}$ integers such that $0 \leq \nu_1 < \cdots < \nu_{\ell}$ and $m_1 + m_2 + \cdots + m_{\ell}=k$. Moreover, let $G(z)= \sum_{i \geq 0} G_i z^i \in \mathbb F[z]^{n \times k}$ with $G_i$ defined in (\ref{Gi}) and $\mathbb{F}=\mathbb{F}_{p^N}$, for $p$ prime and $N$ sufficiently large, so that ${\cal G}(\epsilon_0)$ (defined in (\ref{Gd}), with $\epsilon_0$ defined in (\ref{epsilon0})) satisfy the conditions of theorem \ref{diagonal}. Then ${\cal C} = \mbox{Im}_{\mathbb F[z]}G(z)$ is an optimal $(n,k,\nu_1,m_1)$ convolutional code with Forney indices $\nu_1, \dots, \nu_{\ell}$ with multiplicities $m_1, \dots, m_{\ell}$, respectively.
\end{corollary}

\section{Conclusion}\label{sec:conclusions}

In this paper we have introduced a very general class of superregular matrices and we have shown that these matrices have the property that any combination of its columns have the maximum number of nonzero elements possible for its configuration of zeros. It turns out that this important property can be used to present novel constructions of convolutional codes that attain the maximum possible distance for some fixed parameters of the code, namely, the rate and the Forney indices. These results answered some open questions on distances and constructions of convolutional codes posted in \cite{GRS2006,mceliece98}.

\bibliographystyle{plain}
\bibliography{code1,Ref-Articles,Ref-Climent-2,Ref-PhDT-MT}

\begin{thebibliography}{10}

\bibitem{ANP2013}
P.~Almeida, D.~Napp, and R.~Pinto.
\newblock A new class of superregular matrices and {MDP} convolutional codes.
\newblock {\em Linear Algebra and its Applications}, 439:2145--2157, 2013.

\bibitem{Ando1987}
T.~Ando.
\newblock Totally positive matrices.
\newblock {\em Linear Algebra and its Applications}, 90:165--219, 1987.

\bibitem{CIM1998}
E.~B. Curtis, D.~Ingerman, and J.~A. Morrow.
\newblock Circular planar graphs and resistor networks.
\newblock {\em Linear Algebra and its Applications}, 283:115--150, 1998.

\bibitem{forney75}
G.D. {Forney, Jr.}
\newblock Minimal bases of rational vector spaces, with applications to
  multivariable linear systems.
\newblock {\em SIAM J. Control}, 13:493--520, 1975.

\bibitem{Gan59a}
F.R. Gantmacher.
\newblock {\em The Theory of Matrices}, volume 1,2.
\newblock Chelsea, New York, 1959.

\bibitem{GRS2006}
H.~Gluesing-Luerssen, J.~Rosenthal, and R.~Smarandache.
\newblock Strongly {MDS} convolutional codes.
\newblock {\em {IEEE Trans. Inf. Th}}, 52(2):584--598, 2006.

\bibitem{HuSmTr2008}
R.~Hutchinson, R.~Smarandache, and J.~Trumpf.
\newblock On superregular matrices and {MDP} convolutional codes.
\newblock {\em Linear Algebra and its Applications}, 428:2585--2596, 2008.

\bibitem{JZ1999}
R.~Johannesson and K.~S. Zigangirov.
\newblock {\em Fundamentals of Convolutional Coding}.
\newblock IEEE Press Series in Digital and Mobile Comm., 1999.

\bibitem{mceliece98}
R.J. McEliece.
\newblock The algebraic theory of convolutional codes.
\newblock In R.A.~Brualdi V.S.~Pless, W.C.~Huffman, editor, {\em Handbook of
  Coding Theory Vol. 1}. North-Holland, Amsterdam, 1998.

\bibitem{NaRo2015}
D.~Napp and R.~Smarandache.
\newblock Constructing strongly mds convolutional codes with maximum distance
  profile.
\newblock {\em Advances in Mathematics of Communications}.

\bibitem{Pinkus2009}
A.~Pinkus.
\newblock {\em Totally Positive Matrices}, volume No. 181.
\newblock Cambridge Tracts in Mathematics, 2009.

\bibitem{RosenthalS99}
J.~Rosenthal and R.~Smarandache.
\newblock Maximum distance separable convolutional codes.
\newblock {\em Appl. Algebra Engrg. Comm. Comput}, 10(1):15--32, 1999.

\bibitem{Roth1989}
R.~M. Roth and A.~Lempel.
\newblock On {MDS} codes via {Cauchy} matrices.
\newblock {\em IEEE Trans. Inf. Th}, 35(6):1314--1319, 1989.

\bibitem{Roth1985}
Ron~M. Roth and Gadiel Seroussi.
\newblock On generator matrices of {MDS} codes.
\newblock {\em IEEE Trans. Inf. Th}, 31(6):826--830, 1985.

\bibitem{Smarandache2001}
R.~Smarandache, H.~Gluesing-Luerssen, and J.~Rosenthal.
\newblock Constructions of {MDS}-convolutional codes.
\newblock {\em IEEE Trans. Inf. Th}, 47(5):2045--2049, 2001.

\end{thebibliography}

\end{document}